\documentclass[conference,12pt]{article} 

\usepackage{amsmath}
\usepackage{amssymb}
\usepackage{algorithm}
\usepackage{amsthm}
\usepackage{graphicx}
\usepackage{subfigure}
\usepackage{enumerate}

\newtheorem{thm}{Theorem}[section]
\newtheorem{lem}{Lemma}[section]

\newcommand{\x}{\mathbf{x}}

\newcommand{\f}{\mathbf{f}}
\newcommand{\g}{\mathbf{g}}
\newcommand{\h}{\mathbf{h}}
\newcommand{\X}{\mathbf{X}}

\newcommand{\A}{\mathbf{a}}
\newcommand{\U}{\mathbf{u}}
\newcommand{\eps}{\epsilon}

\newcommand{\eq}{\begin{equation}}
\newcommand{\eqs}{\begin{equation*}}
\newcommand{\eqend}{\end{equation}}
\newcommand{\eqsend}{\end{equation*}}

\begin{document}

\title{Recovery of Sparse 1-D Signals from the Magnitudes of their Fourier Transform}
\author{\begin{tabular}[t]{c@{\extracolsep{5em}}c@{\extracolsep{5em}}c}
Kishore Jaganathan & Samet Oymak & Babak Hassibi\end{tabular}\\
 \\
        Department of Electrical Engineering, Caltech \\
       Pasadena, CA~~91125
}
\date{}
\maketitle

\begin{abstract} 
The problem of signal recovery from the autocorrelation, or equivalently, the magnitudes of the Fourier transform, is of paramount importance in various fields of engineering. In this work,  for one-dimensional signals, we give conditions, which when satisfied, allow unique recovery from the autocorrelation with very high probability. In particular, for sparse signals, we develop two non-iterative recovery algorithms. One of them is based on combinatorial analysis, which we prove can recover signals upto sparsity $o(n^{1/3})$ with very high probability, and the other is developed using a convex optimization based framework, which numerical simulations suggest can recover signals upto sparsity $o(n^{1/2})$ with very high probability.
\end{abstract}

\let\thefootnote\relax\footnotetext{This work was supported in part by the National Science Foundation under grants CCF-0729203, CNS-0932428 and CCF-1018927, by the Office of Naval Research under the MURI grant N00014-08-1-0747, and by Caltech's Lee Center for Advanced Networking.}

\section{Introduction} 

Signal extraction from the autocorrelation, or equivalently, from  the magnitude of the Fourier Transform is known as phase retrieval. This problem fundamentally arises in many practical systems such as X-ray crystallography \cite{millane}, astronomical imaging \cite{dainty}, channel estimation,  speech recognition \cite{rabiner} etc, and has attracted a considerable amount of attention from researchers over the last few decades \cite{walther}. Various algorithms have been proposed to retrieve phase information \cite{gerchberg,gonzalves} and a comprehensive survey of them can be found in \cite{fienup,kailath}.

For one-dimensional signals, since the mapping from signals to autocorrelation is not one-to-one, unique recovery is not possible in general. For any given Fourier transform magnitude,  every possible phase corresponds to a different signal. Hence, additional prior information on the signal is required to limit the number of valid phase combinations. \cite{candes} uses multiple structured illuminations, in which several patterns using different masks are collected to guarantee uniqueness. 

We assume that the signal is sparse, i.e., the number of non-zero entries in the signal is much less compared to the length of the signal. This constraint greatly limits the number of possible phase combinations, and research has been done recently to exploit this feature \cite{vetterli,kishore}. In many applications of phase retrieval, the signals encountered are naturally sparse. For example, astronomical imaging deals with the locations of the stars in the sky, electron microscopy deals with the density of electrons, and so on. 

In this work, we prove that signals can be recovered from their autocorrelation with arbitrarily high probability under certain conditions. We prove this using dimension counting, based on the ideas used in \cite{fannjiang,hayes} for multidimensional signals. We also propose two non-iterative recovery algorithms to extract sparse signals from their autocorrelation. Note that the phase recovery problem is inherently non-convex, and relaxations similar to the ones used in \cite{yonina,kishore, candes2} are used to develop a convex-optimization based framework.

The paper is organized as follows. In Section 2, we discuss some  properties of autocorrelation and spectral factorization which we use for signal extraction. In Section 3, we prove that signals can be recovered from their autocorrelation with very high probability under certain conditions. Non-iterative recovery algorithms are proposed for extraction of the signal from their autocorrelation in Section 4. Section 5 presents the simulation results and concludes the paper.

\section{Theory}

Let $\x=(x_0,x_1,....x_{n-1})$ be a real-valued signal of length $n$. Its autocorrelation, denoted by $\A=(a_0,a_1,....a_{n-1})$, is defined as

\eq
\label{autocorrdef}
a_i \overset{def}{=} \sum_jx_jx_{j+i} = (\x \star \tilde{\x})_i
\eqend
where $\tilde{\x}$ is the time-reversed version of $\x$. Note that cyclic indexing scheme is used in this definition. Rewriting (\ref{autocorrdef}) in the $z$-domain, we get
\eq
\label{autocorrz}
A(z)=X(z)X(z^{-1})
\eqend
where $A(z)$ and $X(z)$ are the $z$-transforms of $\A$ and $\x$ respectively. Since $\x$ is real valued, $X(z)$ is a polynomial in $z$ with real coefficients and hence its zeros occur in conjugate pairs. Also, since $A(z)=A(z^{-1})$, if $z_0$ is a zero of $A(z)$, then $z_0^{-1}$ is also a zero. Hence, the zeros of $A(z)$ appear in quadruples of the form $(z_0,z_0^\star, z_0^{-1}, z_0^{-\star})$. 

The extraction of  $\x$ from $\A$, or equivalently $X(z)$ from $A(z)$, is known as spectral factorization and deals with the distribution of these quadruples between $X(z)$ and $X(z^{-1})$. For every quadruple $(z_0,z_0^\star, z_0^{-1}, z_0^{-\star})$, we can either assign $(z_0,z_0^\star)$ to $X(z)$ and $(z_0^{-1}$,$z_0^{-\star})$ to $X(z^{-1})$, or assign $(z_0^{-1}$,$z_0^{-\star})$  to $X(z)$ and $(z_0,z_0^\star)$ to $X(z^{-1})$. The total number of different valid factorizations hence is exponential in the number of such quadruples. 

\begin{lem}
\label{autolem}
If two distinct finite-length real-valued signals $\f_1$ and $\f_2$ have the same autocorrelation, then there exists finite-length real-valued signals $\g$ and $\h$ such that 
\eq
\f_1=\g \star \h \ \ \ \ \ \ \f_2=\g \star \tilde{\h} 
\eqend
where $\tilde{\h}$ is the time-reversed version of $\h$.
\end{lem}

\begin{proof}

Let $F_1(z)$, $F_2(z)$, $G(z)$ and $H(z)$ be the $z$-transforms of the signals $\f_1$, $\f_2$, $\g$ and $\h$ respectively. Since $\f_1$ and $\f_2$ have the same autocorrelation, (\ref{autocorrz}) gives us

\eq
A(z)=F_1(z)F_1(z^{-1})=F_2(z)F_2(z^{-1})
\eqend
where $A(z)$ is the $z$-transform of the autocorrelation of $\f_1$ and $\f_2$. For every quadruple ($z_0$,$z_0^\star$,$z_0^{-1}$,$z_0^{-\star}$) which are zeros of $A(z)$,  $(z_0,z_0^\star)$ has to be assigned to  $F_1(z)$ or $F_1(z^{-1})$, and $F_2(z)$ or $F_2(z^{-1})$. Let $P_{1}(z)$, $P_{2}(z)$ and $P_3(z)$ be the polynomials constructed from such conjugate pairs of zeros which are assigned to $(F_1(z),F_2(z))$ and  $(F_1(z),F_2(z^{-1}))$ and $(F_1(z^{-1}),F_2(z))$ respectively. Note that $P_2(z)=P_3(z^{-1})$. We have

\eq
F_1(z)=P_1(z)P_2(z) 
\eqend
\eq
 F_2(z)=P_1(z)P_3(z)=P_1(z)P_2(z^{-1})
\eqend
and hence $F_1(z)$ and $F_2(z)$ can be written as
\eq
F_1(z)=G(z)H(z) \ \ \ \ \ \ F_2(z)=G(z)H(z^{-1})
\eqend
where $G(z)=P_1(z)$ and $H(z)=P_2(z)$, or equivalently
\eq
\f_1=\g \star \h \ \ \ \ \ \ \f_2=\g \star \tilde{\h} 
\eqend
in the time domain as the $z$-transform of $\tilde{\h}$ is $H(z^{-1})$.
\end{proof}

\section{Unique Recovery}

In this section, we establish the fact that within the class of signals with non-uniform support (defined later), there is a one-to-one mapping between signals and their autocorrelation almost surely. 

\begin{lem}
\label{hayes}
If $f: \mathcal{A} \rightarrow \mathcal{B}$ is a map from $\mathcal{A}$ to $\mathcal{B}$, where $\mathcal{A}$ is a manifold of dimension $d_a$ and $\mathcal{B}$ is a manifold of dimension $d_b$, then the image of $f$ is  measure zero in $\mathcal{B}$ if $d_a<d_b$.
\end{lem}
Note that any signal of length $n$ can be represented as a vector in $\mathcal{R}^n$. Let $\f$ be a finite-length real-valued signal of length $n$. Let $I$ represent its support,  defined as the set of locations where the $\f$ can have non-zero entries. We say that a signal $\f$ has uniform support if the indices of the elements belonging to the support are periodic, i.e., in an arithmetic progression.  The size of the set $I$ denotes the sparsity of $\f$. Let $\mathcal{F}_k$ denote the set of signals with sparsity $k$. Observe that $\mathcal{F}_k$ is a manifold of dimension $k$.

\begin{lem}
Suppose $\g$ and $\h$ are finite-length real-valued signals with support set $I_g$ and $I_h$ of sparsity $k_g$ and $k_h$ respectively. If $\mathcal{F}_{gh}$ denotes  the set of signals $\g*\h$, and $I_{gh}$ its support. Then
\begin{enumerate}[(i)]
\item The set $\mathcal{F}_{gh}$ is a manifold of dimension $k_g+k_h-1$.
\item $I_{gh}$ has sparsity $k_{gh} \geq k_g+k_h-1$, with equality iff $\g$ and $\h$ have uniform support.
\item If $\f=\g*\h$, where $I$, the support of $\f$, is a subset of $I_{gh}$ with sparsity $k$. The set of such $\f$ is a manifold of dimension $k_g+k_h-1-\gamma$, where $\gamma = k_{gh}-k$.
\end{enumerate}
\label{convprop}
\end{lem}

\begin{proof}
We refer the readers to \cite{fannjiang} for the proof of $(i)$ and $(iii)$. $(ii)$ directly follows from the properties of convolution.
\end{proof}

\begin{lem}
Suppose $\f=\g*\h$, with $\f$ having non-uniform support where as $\g$ and $\h$ have uniform support, also has the additional property that $\f'$ has non-uniform support. Then, the set of such signals is a manifold of dimension strictly lesser than $k_g+k_h-1-\gamma$.
\label{big}
\end{lem}

\begin{proof}
The idea of the proof is  similar to \cite{fannjiang}, based on dimension counting. We saw in Lemma \ref{convprop} that the set of signals $\f$ which can be represented as $\g*\h$ with sparsity $k$ can be written as a manifold of dimension $k_g+k_h-1-\gamma$. The new set of  constraints introduced by terms in $\f'$ being $0$ result in a further reduction in dimension. Hence the set of such signals belong to a manifold of dimension strictly lesser than $k_g+k_h-1-\gamma$.

\end{proof}

\begin{thm}[Main Theorem]
\label{mainthm}
Signals can be uniquely recovered from their autocorrelation, or equivalently, from the magnitudes of their Fourier Transforms almost surely iff they have non-uniform support.
\end{thm}

\begin{proof}
Let $\mathcal{F}'_k$ be the set of all signals $\f$ with non-uniform support of sparsity $k$ which have another signal $\f'$ with non-uniform support and same autocorrelation. Note that $\mathcal{F}_k'$ is the set of signals of sparsity $k$ which cannot be recovered uniquely from their autocorrelation. Lemma \ref{autolem} showed the existence of signals $\g$ and $\h$ such that 
\eq
\f=\g*\h \ \ \ \ \  \ \f'=\g*\tilde{\h}
\eqend
From Lemma \ref{convprop}, we note that the dimension of $\mathcal{F}'_k$ is less than or equal to $k_g+k_h-1-\gamma$

\subsubsection*{Case I: $k_{gh}>k_g+k_h-1$} 

This is the case if $\g$ and $\h$ do not have uniform support. In this case, the dimension of $\mathcal{F}'_k$ is strictly less than $k$. Hence from Lemma \ref{hayes}, we see that $\mathcal{F}'_k$ is a set of measure zero in $\mathcal{F}_k$ and signals with non-uniform support can be recovered from their autocorrelation almost surely.

\subsubsection*{Case II: $k_{gh} = k_g+k_h-1$}

In this case, $\g$ and $\h$ have uniform support. If $\f$ and $\f'$ have non-uniform support, from Lemma \ref{big}, we see that the dimension of $\mathcal{F}_k'$ is strictly lesser than $k=k_g+k_h-1-\gamma$, and hence can be uniquely recovered from their autocorrelation almost surely.

Suppose $\f$ or $\f'$ have uniform support, there will be no additional reduction in dimension.  This case is equivalent to recovering a one-dimensional signal uniquely with no additional constraints, which is almost surely not possible. 
\end{proof}

\section{Recovery Algorithms}

In this section, we develop two non-iterative recovery algorithms for the extraction of sparse signals from their autocorrelation. 

\subsection{Algorithm 1}

Algorithm 1 is based on combinatorial analysis. We propose a method to recover the support of the signal from the support of the autocorrelation, and prove that recovery is possible with very high probability if the sparsity of the signal is $o(n^{1/3})$. Using this support knowledge, we show that signals can be recovered from the autocorrelation with very high probability. 

Suppose $\x$ is a signal of length $n$ such that each element in $\x$ belongs to the support with a probability $\frac{s}{n}$, where $s=n^\alpha, \alpha\leq1$, independent of each other. Let $\A$ denote its autocorrelation, $k$ denote its sparsity and $D=\{d_1,d_2,.....d_k\}$ be the set of indices of the elements belonging to the support.  Also, let $d_{ij}$ be defined as $|d_i-d_j|$ for $(i,j)=\{1,2,....k\}$. If $A$ is the set of indices of  elements belonging to the support of the autocorrelation, then $A=\{\bigcup_{i,j}{d_{ij}}\}$. Note that $d_{i,i+1}$ is a geometric random variable with parameter $\frac{s}{n}$. Without loss of generality, let us assume $d_{k-1,k} \geq d_{12}$, otherwise we could just flip the signal and consider the flipped signal. Define $A_1=\{d_{ij}-d_{12}|d_{ij} \in A\}$ and $A_2=\{d_{ij}-d_{k-1,k}|d_{ij} \in A\}$. 

The algorithm for signal recovery is described below. In what follows, we give a sequence of lemmas to justify various steps of the algorithm. 

\begin{algorithm}[h]
\caption{}
\label{algo1}
 \textbf{Input:} The autocorrelation $\bf{a}$ of the signal. \\
 \textbf{Output:} The sparse signal ${\bf{x}}$ which has  autocorrelation $\bf{a}$

\begin{itemize}
\item Extract $d_{12}$ and $d_{k-1,k}$ from $A$ (Lemma \ref{extractd}). Calculate the sets $A_1$ and $A_2$.
\item Perform $(A\cap A_1) \cap (d_{2,k-1}-(A\cap A_2))$ and identify the support $\U$ of $\x$ (Lemma \ref{extractu})
\item Construct the graph $G$ (Lemma \ref{defg}) using $\U$, identify an odd cycle and a path connecting all the vertices and extract $\x$. 
\end{itemize}

\end{algorithm}

\begin{lem}
The sparsity $k$ of the signal satisfies $(1-\eps)s \leq k  \leq (1+\eps)s$ with very high probability for any $\eps>0$, $n>n(\eps)$.
\end{lem}

\begin{proof}
Use Chebyshev's inequality.
\end{proof}

\begin{lem} For three independent random variables $X_1$, $X_2$ and $X_3$ where $X_1$ and $X_2$ are geometric random variables with parameter $\frac{s}{n}$ , $P(X_1-pX_2=qX_3) \leq \frac{s}{n}$ if $s=n^\alpha, \alpha<1$ for  $n>n(\eps)$, where $p$ and $q$ are integers.
\label{lem41}
\end{lem}

\begin{proof}
Refer Appendix.
\end{proof}

\begin{lem}
\label{lem43}
$P(d_{k-1,k}-d_{12} \in A) \leq (1+\eps)\frac{s^3}{n}$ for any $\eps>0$, $n>n(\eps)$.
\end{lem}

\begin{proof}
Using union bound, we obtain
\eq
P(d_{k-1,k}-d_{12} \in A) \leq \sum_{i}\sum_{j} P(d_{k-1,k}-d_{12}= d_{ij}) 
\eqend 
\eqs
= \sum_{i \neq1}\sum_{j\neq k} P(d_{k-1,k}-d_{12}= d_{ij})+ \sum_{i\neq1}P(d_{k-1,k}-d_{12} = d_{ik} )
\eqsend
\eq
+\sum_{j \neq k}P(d_{k-1,k}-d_{12} = d_{1j}))
\eqend
Note that the $d_{ij}$'s for $i\neq1,j\neq k$ are independent of $d_{12}$ and $d_{k-1,k}$. Hence Lemma \ref{lem41} can be applied and each term in the first summation can be upper bounded by $\frac{s}{n}$. Since $d_{k-1,k} < d_{ik}$ and $d_{12}>0$, all the terms in the second summation are  zero. The terms in the third summation can be equivalently written as $P(d_{k-1,k}-2d_{12}= d_{2j})$, and Lemma \ref{lem41} can be used to upper bound every term by $\frac{s}{n}$. Since $d_{1k}$ is the largest sum, we need not consider it in the summation. Hence, we get
\eq
P(d_{k-1,k}-d_{12} \in A)\leq k^2{\frac{s}{n}} \leq (1+\eps')^2\frac{s^3}{n} \leq (1+\eps)\frac{s^3}{n}
\eqend
\end{proof}

\begin{lem}
\label{extractd}
$d_{12}$ and $d_{k-1,k}$ can be recovered from the autocorrelation with very high probability if $s=o(n^{1/3})$.
\end{lem}

\begin{proof}
The first and second highest terms in $A$ are $d_{1k}$ and $d_{2k}$ respectively since  $d_{12}\leq d_{k-1,k}$. Note that $d_{1k}-d_{2k} = d_{12}$, hence $d_{12}$ can be recovered from the autocorrelation. The only terms that can be higher than $d_{1,k-1}$ in $A$ are $\{d_{3k}, d_{4k},.....d_{k-1,k}\}$. Note that $d_{2k}-d_{ik}=d_{2i}$, which belongs to $A$ for all  $i=\{3,.....k-1\}$. So if $d_{2k}-d_{1,k-1}$ doesn't belong to $A$, we can recover $d_{1,k-1}$ by considering the highest term which when subtracted from $d_{2k}$ produces a value which doesn't belong to $A$. The probability of failure can hence be written as $P(d_{k-1,k}-d_{12} \in A)$ which goes to zero if $s=o(n^\frac{1}{3})$, as seen in Lemma \ref{lem43}. Hence both $d_{12}$ and $d_{k-1,k}$ can be recovered with very high probability if $s=o(n^{1/3})$.
\end{proof}

With the knowledge of $d_{12}$ and $d_{k-1,k}$, we can construct the sets $A_1$ and $A_2$. Consider the intersection of $A$ and $A_1$. All entries of the form $d_{2i}$ for $i=\{3,4,...k\}$ will survive trivially for any signal. Similarly, all entries of the form  $d_{i,k-1}$ for $i=\{1....k-2\}$ will survive the intersection of $A$ and $A_2$ for any signal. If we subtract the survivors of the intersection of $A$ and $A_2$ from $d_{2,k-1}$, we get $d_{2i}$ for $i=\{3,4,...k-1\}$. Hence the elements $d_{2i}$ for $i=\{3,4,...k-1\}$ will survive $(A\cap A_1) \cap (d_{2,k-1}-(A\cap A_2))$.

\begin{lem}
\label{extractu}
No other $d_{ij}$ will survive $(A\cap A_1) \cap (d_{2,k-1}-(A\cap A_2))$ and hence the support can be recovered with very high probability if $s=o(n^{1/3})$ 
\end{lem}

\begin{proof}
Suppose you choose $d_{ij}$ such that $i$ and $j$ are picked at random. The probability that $d_{ij}$ is a particular value can be upper bounded by $\frac{1}{n}$. For a non-trivial $d_{ij}$ in $A$ to survive  $A \bigcap A_1$, $d_{ij}+d_{12}$ has to be in $A$. Similarly, $d_{2k}-d_{ij}$ and $d_{2,k-1}-d_{ij}$ has to be in A for it to survive $d_{2k}-A \bigcap A_2$. Using union bounds, we see that the probability of survival of some other $d_{ij}$  goes to $0$ when $s=o(n^{1/3})$. Note that we have information about $d_{k-1,k}$ upto  $s=o(n^{1/3})$.

If no other elements survive, from $d_{2i}$ for $i=\{3,4,...k-1\}$, we can extract $d_{i,i+1}$ for  $i=\{3,4,...k-2\}$ and since we already know $d_{12}$ and $d_{k-1,k}$, we have the support of the signal.

\end{proof}

Suppose we have the support of the signal, $D=\{d_1,d_2,.....d_k\}$ being the indices of the elements belonging to the support. Define a pair $(d_i,d_j)$ as a good pair if they are the only pair separated by $|d_i-d_j|$. Note that for such a pair, $a_{|d_i-d_j|}=x_{d_i}x_{d_j}$

\begin{lem}
\label{defg}
Consider a graph $G$ with $k$ vertices, each vertex representing an element of the support. Draw a weighted edge between every good pair, the weight being the value of the corresponding autocorrelation. If the graph $G$ has an odd cycle and is connected, then the signal can be extracted from the autocorrelation upto a global sign.
\end{lem}

Consider an odd cycle with $2r-1$ vertices ${i_1},{i_2},...i_{2r-1}$. The term $\frac{x_{i_1i_2}x_{i_3i_4}....x_{i_{2r-1}i_1}}{x_{i_{2}i_{3}}...x_{i_{2r-2}i_{2r-1}}}$ gives $x_{i_1}^2$, from which $x_{i_1}$ can be extracted upto a sign, and from it the other terms in the odd cycle can be extracted using the weight corresponding to the edges. Since the graph is connected, all the other terms can be calculated.

\begin{lem}
The graph $G$ has an odd cycle and is connected with very high probability for $s=o(n^{1/3})$.
\end{lem}
Pick any three vertices randomly. Choose any path of length $k-3$ from one of those vertices to cover all the remaining vertices randomly. If all the edges exists between the three vertices and the chosen $k-3$ length path exists, we are through. If any of the $k$ edges doesn't exist, it implies that the distance between that pair of vertices occurs more than once. Since there are less than $k^2$ pairs,  the probability of a pair of vertices not having an edge can be union bounded by $\frac{k^2}{n}$. Since there are $k$ edges to be considered, the  probability of  failure can be upper bounded by $\frac{k^3}{n}$. Hence if $s=o(n^{1/3})$, any chosen triangle and path exists with very high probability. 

\subsection{Algorithm 2}

Algorithm 2 is developed using a convex optimization based framework. Semidefinite relaxation is used to convert the non-convex constraints into a set of convex constraints. We break the problem into two stages. First, the support of the signal is recovered from the autocorrelation and then we solve for the signal in the support. 

\subsubsection{Support Recovery}

We have to extract  $\U$ from the autocorrelation of the signal. We will assume that the support of the signal is a subset of the support of the autocorrelation. This is the same as assuming there is no cancellation of support in the autocorrelation, which is a very weak requirement and holds with probability one if the coefficients of the signal are chosen randomly from a non-degenerate distribution. With this assumption, $a_i=0$ implies that no two elements in the support are separated by a distance $i$, and if $a_i$ is non-zero, there is atleast one pair of elements in the support separated by a distance $i$, i.e.,

\eq
a_i=0 \Rightarrow u_ju_{i+j}=0 \ \forall \ j
\eqend
\eq
a_i \neq 0 \Rightarrow u_ju_{i+j} \neq 0 \ \textrm{for some} \ j
\eqend
where $\U$ is the binary support vector. This is clearly non-convex as the  constraints are non-convex and $\U$ is binary. Define $\mathbf{S}=\U\U^T$, which is allowed to be positive semidefinite, as it is the smallest convex set containing all rank one matrices. The entries of $\mathbf{S}$ are allowed to be in $[0,1]$, which is the best convex relaxation for binary variables. The trace of $\mathbf{S}$ is given by $\sum_i{u_i^2}=\sum_iu_i=k$, the sparsity of the signal. Also, note that $\sum_iS_{ij}=\sum_iu_iu_j=u_j\sum_iu_i=ku_j=ku_j^2=kS_{jj}$ and similarly $\sum_jS_{ij}=kS_{ii}$. Since flipped version of the support also satisfies all the constraints, a random matrix  $\mathbf{V}$ is used to bias the cost. The support estimation problem becomes
\begin{align}
\nonumber &\textrm{minimize}
& &{trace }\mathbf{(VS)} \\
\nonumber & \textrm{subject to   } 
\nonumber & & {trace } (\mathbf{S})=k \quad \quad \mathbf{S} \succcurlyeq 0 \\
\nonumber & & &\sum_{i}S_{ij}=kS_{jj} \quad \sum_{j}S_{ij}=kS_{ii}\\
\nonumber & & &\sum_{i}S_{i,i+k}>0 \textrm{ iff } a_k \neq 0 \\
\nonumber & & & 0 \leq S_{ij} \leq 1 \quad  0 \leq i,j \leq m-1 \\
\end{align}

Note that we assume apriori knowledge of the sparsity of the signal, i.e., the number of non-zero locations of the signal is known.

\subsubsection{Signal Recovery}

Note that the autocorrelation constraints are non-convex. As we did in the support extraction, we use the semidefinite relaxation $\mathbf{X}=\x\x^T$. We append $n$ zeros to the signal so that cyclic indexing scheme can be applied, hence a $m=2n$ order DFT matrix is required. Suppose $\mathbf{M_k}$ is the $m\times m$ matrix defined by $\mathbf{M}_k=\mathbf{f}_k\mathbf{f}_k^T$, where $\mathbf{f}_k$ is the $k^{th}$ column of the DFT matrix for $k=\{0,1,....m-1\}$. The autocorrelation constraints can be written in the Fourier domain as
\eq
{\mathbf{Y}_k}={trace}({\bf{M_kX}})   \quad \quad  k=0,......,m-1
\eqend 
where $\mathbf{Y}=\{|y_0|^2,|y_1|^2,......|y_{m-1}|^2\}$ is the vector containing the squared magnitude of the Fourier transform of $\x$. We can solve for the signal using L1-minimization \cite{candes3, candes4,candes5}.
\begin{align}
\nonumber &\textrm{minimize}
& & ||\mathbf{X}||_1 \\
& \textrm{subject to   } &  & {\bf{Y_k}}={trace }({\bf{M_kX}}),    \quad  k=0,......m-1\\
\nonumber & & & X_{ij}=0 \ \ \ \textrm{if} \ \ \  S_{ij} = 0  \quad  0 \leq i,j \leq m-1 \\
\nonumber & & & \X \succcurlyeq 0
\end{align}

\section{Simulation Results}

\begin{figure}[h]
\begin{center}
\includegraphics[scale=0.65]{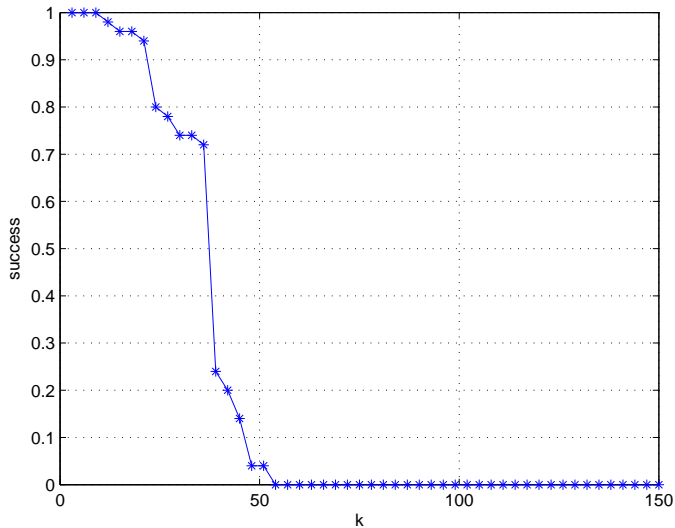}  
\caption{Success rate of recovery using Algorithm 1 for $n=8192$ $(n^{1/3}\approx 20)$ for various sparsities}
\end{center}
\end{figure}
\begin{figure}[h]
\begin{center}
\includegraphics[scale=0.65]{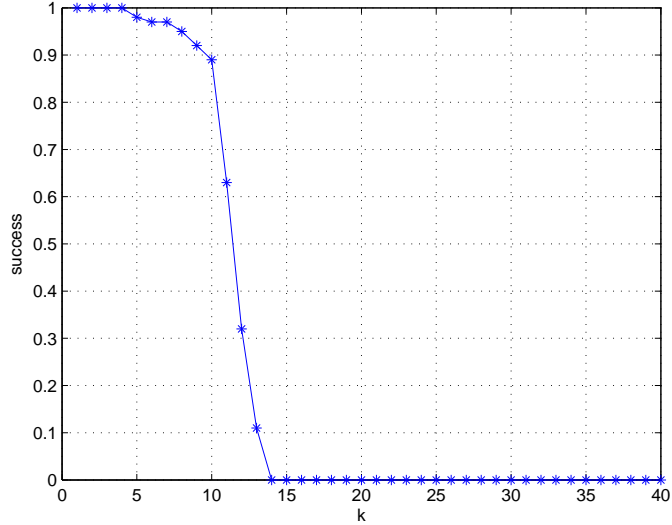}  
\end{center}
\caption{Success rate of recovery using Algorithm 2 for $n=64$ $(n^{1/2}=8)$ for various sparsities}
\label{fig1}
\end{figure}

Figure 1 shows the success rate of signal recovery using Algorithm 1 as a function of the sparsity of the signal. We see that signals with $s=o(n^{1/3})$ are recovered successfully with very high probability. While the algorithm is computationally very cheap, it is not robust to noise due to error propagation.

Figure 2 demonstrates the performance of Algorithm 2 as a function of the sparsity of the signal. Numerical simulations strongly suggest that signals with sparsity upto $s=o(n^{1/2})$ can be recovered using this algorithm. It is also very robust to noise and hence more practical. We hope to provide theoretical guarantees in a future publication.

\section{Appendix}
\begin{lem}
\label{geom}
For a pair of geometric random variables $X_1$ and $X_2$ with parameter $\frac{s}{n}$ each, $P(X_1-pX_2=c) \leq \frac{s}{n}$ if $s=n^\alpha, \alpha<1$ for $n>n(\eps)$, where $p$ and $c$ are integers.
\end{lem}
\subsection{Proof of Lemma \ref{geom}}
\eq
P(X_1-pX_2=c)=\sum_{i=0}^{\infty}P(X_2=i)P(X_1=pi+c) 
\eqend
\eqs
= \sum_{i=0}^{\infty} (1-\frac{s}{n})^{i}\frac{s}{n}(1-\frac{s}{n})^{c+pi}\frac{s}{n}= (\frac{s}{n})^2(1-\frac{s}{n})^{c}\sum_{i=0}^{\infty}(1-\frac{s}{n})^{(1+p)i}
\eqsend
\eqs
= (\frac{s}{n})^2 \frac{(1-\frac{s}{n})^c}{1-(1-\frac{s}{n})^{(1+p)}} =  (\frac{s}{n})^2 \frac{1}{(1+p)\frac{s}{n}+\frac{s}{n}o(1)} \leq  \frac{s}{n}
\eqsend
for $n>n(\eps)$.

\subsection{Proof of Corollary \ref{lem41}}
From Lemma \ref{geom}, we see that
\eqs
P(X_1-pX_2=qX_3) =\sum_{i=0}^{\infty}P(X_3=i)P(X_1-pX_2=qi)
\eqsend
\eq
\leq \sum_{i=0}^{\infty}P(X_3=i)\frac{s}{n} \leq \frac{s}{n}\sum_{i=0}^{\infty}P(X_3=i) \leq \frac{s}{n}
\eqend

\end{document}